\documentclass{article}
\usepackage[utf8]{inputenc}
\usepackage{csquotes}

\usepackage{amssymb}
\usepackage{amsmath}

\usepackage{amsthm}
\theoremstyle{plain} \newtheorem{defi}{Definition}
\theoremstyle{plain} \newtheorem{prop}{Proposition}
\theoremstyle{plain} \newtheorem{teo}[prop]{Theorem}
\theoremstyle{plain} 
\theoremstyle{plain} \newtheorem{eje}{Example}
\theoremstyle{remark} \newtheorem*{nota}{Notation}
\theoremstyle{plain} \newtheorem*{obse}{Remark}
\theoremstyle{plain} 
\theoremstyle{plain} \newtheorem{procedure}{Procedure}

\usepackage{geometry}
\usepackage{tikz}

\usepackage{hyperref}

\usepackage{comment}

\title{Market viability and completeness for multinomial models}
\author{Nahuel I. Arca}

\begin{document}

\maketitle

\begin{abstract}
    In this paper we aim to study viability and completeness in finite markets. In order to do that, we characterize the set of equivalent martingale measures of two-period markets as convex combinations of a finite number of martingale measures. We provide an algorithm for finding such measures, that can be applied in other problems of convex geometry, and represents the starting point for a study of such characterizations of convex sets' intersections. We apply these results to the study of a discrete-time version of the Korn-Kreer-Lenssen model, and give an example of the limitations of using discrete-time models to understand continuous-time ones.
\end{abstract}

\section{Introduction}

The starting point to the subject of mathematical finance tends to be the binomial model \cite{cox1979option,rendleman1979two,sharpe1999investment}. In such model, time is discrete and, at each instant of time, the future bifurcates: either the stock price goes up by a factor $u>1$, or the stock price goes down by a factor $d<1$. This model is shown in figure \ref{fig:binomial_model}.
\begin{figure}[h]
    \centering
\begin{tikzpicture}
    \node (0) at (0,0) {$S(0)$};
    \node (0u) at (1*2,1*0.5) {$uS(0)$};
    \node (0d) at (1*2,-1*0.5) {$dS(0)$};
    \node (0uu) at (2*2,2*0.5) {$u^2S(0)$};
    \node (0ud) at (2*2,0*0.5) {$udS(0)$};
    \node (0dd) at (2*2,-2*0.5) {$d^2S(0)$};

    \path [->] (0) edge (0u);
    \path [->] (0) edge (0d);
    \path [->] (0u) edge (0uu);
    \path [->] (0u) edge (0ud);
    \path [->] (0d) edge (0ud);
    \path [->] (0d) edge (0dd);

    \node (t0) at (0,2*0.5+1) {$t=0$};
    \node (t1) at (1*2,2*0.5+1) {$t=1$};
    \node (t2) at (2*2,2*0.5+1) {$t=2$};  
\end{tikzpicture}
    \caption{Binomial model with 3 trading dates.}
    \label{fig:binomial_model}
\end{figure}
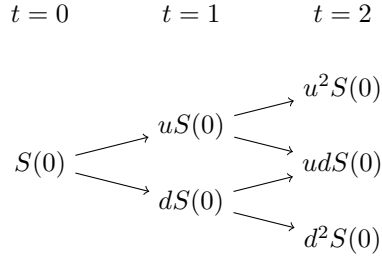

Another variant of this model is the trinomial tree \cite{boyle1986option,clifford2010pricing,haug2007complete,hull2003options,rubinstein2000relation,rubinstein1999rubinstein}. In this model, time is still discrete but, at each instant of time, the future branches in three possibilities: those two of the binomial model, and a third one consisting in the price remaining constant. A version of this model with $ud=1$ is shown in figure \ref{fig:trinomial_model}.
\begin{figure}[h]
    \centering
\begin{tikzpicture}
    \node (0) at (0,0) {$S(0)$};
    \node (0u) at (1*2,1) {$uS(0)$};
    \node (00) at (1*2,0) {$S(0)$};
    \node (0d) at (1*2,-1) {$dS(0)$};
    \node (0uu) at (2*2,2) {$u^2S(0)$};
    \node (0u0) at (2*2,1) {$uS(0)$};
    \node (0ud) at (2*2,0) {$S(0)$};
    \node (0d0) at (2*2,-1) {$dS(0)$};
    \node (0dd) at (2*2,-2) {$d^2S(0)$};

    \path [->] (0) edge (0u);
    \path [->] (0) edge (00);
    \path [->] (0) edge (0d);
    \path [->] (0u) edge (0uu);
    \path [->] (0u) edge (0u0);
    \path [->] (0u) edge (0ud);
    \path [->] (00) edge (0u0);
    \path [->] (00) edge (0ud);
    \path [->] (00) edge (0d0);
    \path [->] (0d) edge (0ud);
    \path [->] (0d) edge (0d0);
    \path [->] (0d) edge (0dd);

    \node (t0) at (0,2+1) {$t=0$};
    \node (t1) at (1*2,2+1) {$t=1$};
    \node (t2) at (2*2,2+1) {$t=2$};  
\end{tikzpicture}
    \caption{Trinomial tree model with 3 trading dates and $ud=1$.}
    \label{fig:trinomial_model}
\end{figure}
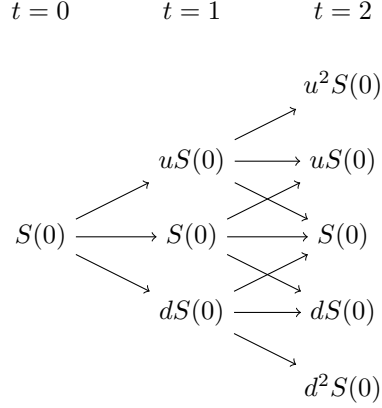

In each one of these models, each path starting at time $t=0$ and finishing at time $t=2$ represents a possible world. In the binomial model, there are four possible worlds: the elements of the sample space $\Omega=\{UU,UD,DU,DD\}$. The information disposable at time $t=i$ is represented by the partition $\mathcal{P}_i$, and we have $\mathcal{P}_0=\{\Omega\}$, $\mathcal{P}_1=\{\{UU,UD\},\{UD,DD\}\}$ and $\mathcal{P}_2=\{\{UU\},\{UD\},\{UD\},\{DD\}\}$. All of these can be represented by the rooted tree depicted in figure \ref{fig:info_binomial}.
\begin{figure}[h]
    \centering
\begin{tikzpicture}
    \node (0) at (0,0) {\textbullet};
    \node (0u) at (1*2,1) {\textbullet};
    \node (0d) at (1*2,-1) {\textbullet};
    \node (0uu) at (2*2,1.5) {$UU$};
    \node (0ud) at (2*2,0.5) {$UD$};
    \node (0du) at (2*2,-0.5) {$DU$};
    \node (0dd) at (2*2,-1.5) {$DD$};

    \path [->] (0) edge (0u);
    \path [->] (0) edge (0d);
    \path [->] (0u) edge (0uu);
    \path [->] (0u) edge (0ud);
    \path [->] (0d) edge (0du);
    \path [->] (0d) edge (0dd);

    \node (t0) at (0,2*0.5+1) {$t=0$};
    \node (t1) at (1*2,2*0.5+1) {$t=1$};
    \node (t2) at (2*2,2*0.5+1) {$t=2$};  
\end{tikzpicture}
    \caption{Information tree of the binomial model with 3 trading dates.}
    \label{fig:info_binomial}
\end{figure}
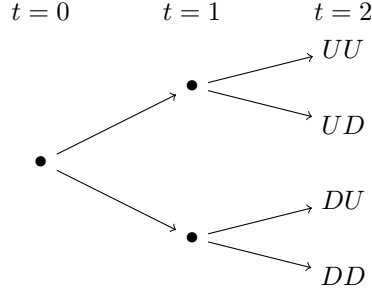

In the same manner, the information in the trinomial tree model can be represented by the rooted tree depicted in figure \ref{fig:info_trinomial}.
\begin{figure}[h]
    \centering
\begin{tikzpicture}
    \node (0) at (0,0) {\textbullet};
    \node (0u) at (1*2,1.5) {\textbullet};
    \node (00) at (1*2,0) {\textbullet};
    \node (0d) at (1*2,-1.5) {\textbullet};
    \node (0uu) at (2*2,2) {$UU$};
    \node (0u0) at (2*2,1.5) {$UI$};
    \node (0ud) at (2*2,1) {$UD$};
    \node (00u) at (2*2,0.5) {$IU$};
    \node (000) at (2*2,0) {$II$};
    \node (00d) at (2*2,-0.5) {$ID$};
    \node (0du) at (2*2,-1) {$DU$};
    \node (0d0) at (2*2,-1.5) {$DI$};
    \node (0dd) at (2*2,-2) {$DD$};

    \path [->] (0) edge (0u);
    \path [->] (0) edge (00);
    \path [->] (0) edge (0d);
    \path [->] (0u) edge (0uu);
    \path [->] (0u) edge (0u0);
    \path [->] (0u) edge (0ud);
    \path [->] (00) edge (00u);
    \path [->] (00) edge (000);
    \path [->] (00) edge (00d);
    \path [->] (0d) edge (0du);
    \path [->] (0d) edge (0d0);
    \path [->] (0d) edge (0dd);

    \node (t0) at (0,2.5) {$t=0$};
    \node (t1) at (1*2,2.5) {$t=1$};
    \node (t2) at (2*2,2.5) {$t=2$};  
\end{tikzpicture}
    \caption{Information tree of the trinomial tree model with 3 trading dates.}
    \label{fig:info_trinomial}
\end{figure}
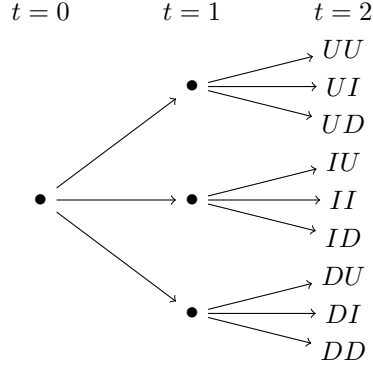

The natural generalization of these examples is any model where the information is represented by a finite rooted tree. These paths of possible futures originate in the Arrow-Debreu model of general equilibrium \cite{arrow1954existence,breeden1978prices,mckenzie1954equilibrium}. In this work, we study the conditions of market viability and completeness, for such models.

In section \ref{teo}, we give a geometric interpretation of the set of equivalent martingale measures and use Krein-Milman's theorem to characterize such set as a convex combination of a finite number of martingale measures. More importantly, we provide an algorithm for computing these generators, and another one for completing arbitrage-free markets in every possible way.

In section \ref{eje}, we show how to apply the theoretical results developed in section \ref{teo} to concrete cases.

In section \ref{apli}, we use the results of section \ref{teo} to prove some results of multinomial models. We analyze a market with a single asset and find a necessary and sufficient condition for it to be arbitrage-free. We also prove that the market is complete when the branching factor $b$ equals $2$. In the case $b=3$, we add a derivative and characterize all the derivatives that complete the market. We apply this last case to the analysis of a discrete-time version of the Korn-Kreer-Lenssen model \cite{korn1998pricing}, and give an example of complete finite markets that converge to a market with arbitrage opportunities.

\section{Theory}\label{teo}

Let's recall some definitions.

A trading strategy is called \emph{self financing} if no money goes in or out in the whole process. In a market with terminal date $T$, an \emph{arbitrage opportunity} is a self financing trading strategy such that the value $V$ of the portfolio satisfies
\begin{align*}
    V(0)&=0\text{ ,}\\
    \mathbb{P}(V(T)\geq 0)&=1\\
    \text{and}\quad\mathbb{P}(V(T)>0)&>0\text{ .}
\end{align*}
A market is \emph{arbitrage-free} (or \emph{viable}) if it doesn't admit any arbitrage opportunities. For the purposes of this paper, a market is \emph{complete} if it is arbitrage-free and for every random variable $X$ there is a trading strategy such that the terminal value $V(T)$ of the portfolio satisfies $V(T)=X$ with probability $1$.

Following \cite{bjork2009arbitrage} but with different notation, let's consider a market model with trading dates $t=0,1$, where the sample space is $\Omega=\{1,2,\ldots,b\}$ with $\mathbb{P}(\{\omega\})>0$ for all $1\leq\omega\leq b$, financial assets $S_0,\ldots, S_n$, and $S_0(\omega,t)>0$ for all $(\omega,t)\in\Omega\times\{0,1\}$.

As in \cite{bjork2009arbitrage}, we define $Z_i:=S_i/S_0$,
\begin{align*}
    D&=\begin{pmatrix}
            S_0(1,1)&\cdots &S_0(b,1)\\
            \vdots&&\vdots\\
            S_n(1,1)&\cdots &S_n(b,1)
        \end{pmatrix}\qquad\text{and}\\
    D^Z&=\begin{pmatrix}
            Z_0(1,1)&\cdots &Z_0(b,1)\\
            \vdots&&\vdots\\
            Z_n(1,1)&\cdots &Z_n(b,1)
        \end{pmatrix}\text{ .}
\end{align*}

To study the way in which arbitrage-free markets can be completed, we will use two results from \cite{bjork2009arbitrage}:
\begin{prop}\label{arbitrage-free}
    The market is arbitrage-free if and only if there exists $q\in\mathbb{R}_{>0}^b$ such that
    \begin{align}
        D^Zq=Z(0)\text{ .}\label{bjork}
    \end{align}
\end{prop}
\begin{prop}\label{complete}
    Let's assume the market is arbitrage-free. The market is complete if and only if $D$ has rank $b$.
\end{prop}

Bearing in mind the general theory of financial markets, if $q\in\mathbb{R}_{\geq 0}^b$ is a solution of \eqref{bjork} we will call it a \emph{martingale measure} (MM), and if $q\in\mathbb{R}_{>0}^b$ we will call it an \emph{equivalent martingale measure} (EMM). Following \cite{delbaen2006mathematics}, we will denote by $\mathcal{M}^e(S)$ the set of EMM's and by $\mathcal{M}^a(S)$ the set of all MM's. Observe that $\mathcal{M}^e(S)\subset \mathcal{M}^a(S)\subset\mathbb{R}^b$ and
\begin{align}
    \mathcal{M}^e(S)=\mathcal{M}^a(S)\cap\mathbb{R}_{>0}^b\text{ .}\label{equiv_mart}
\end{align}

Completing a market consists in adding assets $S_{n+1},\ldots,S_m$ such that the new market is complete. Through proposition \ref{arbitrage-free}, this requires that the original market be arbitrage-free. With the insights of propositions \ref{arbitrage-free} and \ref{complete}, the following procedure for completing an arbitrage-free market becomes clear:
\begin{procedure}
    \begin{enumerate}
        \item Fix an EMM $p$.
        \item Add rows $(S_k(\omega,1))_{k,\omega}$ for $k>n$ to the matrix $D$ until its row rank is $b$.
        \item Define
        \begin{align*}
            S_k(0):=S_0(0)\langle (Z_k(\omega,1))_{\omega},p\rangle
        \end{align*}
        for all the new rows of $D$.
    \end{enumerate}
\end{procedure}
Therefore, the ways in which the market can be completed are constrained by the set $\mathcal{M}^e(S)$. As pointed out by equation \eqref{equiv_mart}, the characterization of the set $\mathcal{M}^a(S)$ leads automatically to the characterization of the set $\mathcal{M}^e(S)$. The following subsection is about the characterization of $\mathcal{M}^a(S)$.

\subsection{Classification of martingale measures}

Observe that an MM $p$ is a solution of
\begin{align}
    \begin{pmatrix}
            Z_1(1,1)&\cdots &Z_1(b,1)\\
            \vdots&&\vdots\\
            Z_n(1,1)&\cdots &Z_n(b,1)
        \end{pmatrix}p=\begin{pmatrix}
            Z_1(0)\\
            \vdots\\
            Z_n(0)
        \end{pmatrix}\label{mart_sign}
\end{align}
in the standard $b-1$-simplex $\Delta^{b-1}\subset\mathbb{R}^b$. We will call $A(S)$ the affine space defined by the solutions of this system (possibly $A(S)=\emptyset$). Therefore,
\begin{align}
    \mathcal{M}^a(S)=A(S)\cap \Delta^{b-1}\text{ .}\label{mart}
\end{align}

It is immediate from equation \eqref{mart} that $\mathcal{M}^a(S)\subset\mathbb{R}^b$ is a compact convex set. A characterization of such sets is given by the Krein-Milman theorem:
\begin{teo}[Krein-Milman]
    A compact convex subset of a Hausdorff locally convex topological vector space is equal to the closed convex hull of its extreme points.
\end{teo}
As a matter of fact, we only need the particular case in which the space considered is $\mathbb{R}^b$ (this result is presented as the Minkowski-Carathéodory theorem in page 126 of \cite{simon2011convexity}). Let's recall the definition of \emph{extreme point} \cite{narici2010topological}.
\begin{nota}
    Given distinct $y,z\in\mathbb{R}^b$,
    \begin{align*}
        (y,z):=\{ty+(1-t)z:t\in (0,1)\}\text{ .}
    \end{align*}
\end{nota}
\begin{defi}
    Let $K$ be a convex set. A convex subset $F$ of $K$ is a \emph{face} of $K$ if for any $x\in F$, $x\in (y,z)$ for $y, z \in K$ implies $y,z \in F$. $x$ is an \emph{extreme point} of $K$ if $\{x\}$ is a face of $K$.
\end{defi}
To find the extreme points of $\mathcal{M}^a(S)$ we will use the following property of faces \cite{weis2024note}:
\begin{prop}
    Let $K,L\subset\mathbb{R}^b$ be convex sets and let $F$ be a nonempty face of $K\cap L$. Then $F$ is the intersection of a face of $K$ and a face of $L$.
\end{prop}
The reciprocal follows immediately from the definition of face. Hence, the faces of $\mathcal{M}^a(S)$ are the intersections of the faces of $\Delta^{b-1}$ with the faces of $A(S)$.

\begin{obse}
    The faces of $\Delta^{b-1}$ are given by
    \begin{align*}
        \text{ch}(\{e_i:i\in\ I\})
    \end{align*}
    where $\text{ch}$ denotes the convex hull and $I\subset\{1,\ldots,b\}$ (this is a one-to-one correspondence). The faces of $A(S)$ are $\emptyset$ and $A(S)$.
\end{obse}

\begin{obse}
    If $F$ and $G$ are faces of $\Delta^{b-1}$ such that $F\cap A(S)$ is a singleton and $G\subset F$, then $G\cap A(S)\in\{\emptyset,F\cap A(S)\}$.
\end{obse}

From these remarks, it follows an algorithm to find the extreme points of $\Delta^{b-1}\cap A(S)$:
\begin{procedure}
\begin{enumerate}
    \item We start by checking if $e_i\in A(S)$ for all $1\leq i\leq b$: we put those that belong to $A(S)$ in a list of generators, and those that don't belong to $A(S)$ in a list of 0-simplices.
    \item We build all the possible 1-simplices whose boundaries are all in the list of 0-simplices. For each one of these 1-simplices, we check if they intersect $A(S)$: we add the intersections thus obtained to the list of generators and, with the 1-simplices that don't intersect $A(S)$ we make a list of 1-simplices.
    \item We build all the possible 2-simplices whose boundaries are all in the list of 1-simplices. For each one of these 2-simplices, we check if they intersect $A(S)$: we add the intersections thus obtained to the list of generators and, with the 2-simplices that don't intersect $A(S)$ we make a list of 2-simplices, etc.
\end{enumerate}
\end{procedure}

Because $\Delta^{b-1}$ doesn't contain $b$-simplices, the process stops in step $b$ or before. A Python implementation of this algorithm can be found in \url{https://github.com/nahueliarca/multinomial/tree/main}.

\begin{eje}
To illustrate the procedure, we will apply it to figure \ref{fig:algorithm}, a case where $b=3$. Such figure represents the plane that contains $\Delta^2$ ($\{x_1+x_2+x_3=1\}$). In this particular case, applying the algorithm results in the following steps:
\begin{enumerate}
    \item We check that $e_i\not\in A(S)$ for $i=1,2,3$. Now, the list of generators is $\{\}$ and the list of $0$-simplices is $\{e_1,e_2,e_3\}$.
    \item We build all the possible 1-simplices whose boundaries are all in the list of 0-simplices, namely $e_1e_2$, $e_1e_3$ and $e_2e_3$. $e_1e_2\cap A(S)=\{\}$, $e_1e_3\cap A(S)=\{p^1\}$ and $e_2e_3\cap A(S)=\{p^2\}$. Now, the list of generators is $\{p^1,p^2\}$ and the list of $1$-simplices is $\{e_1e_2\}$.
    \item We build all the possible 2-simplices whose boundaries are all in the list of 1-simplices: there are none. Therefore, the algorithm ends here and the list of generators is $\{p^1,p^2\}$.
\end{enumerate}
\begin{figure}[h]
    \centering
    \begin{tikzpicture}
        \filldraw[fill=gray!30] (0,0) node[anchor=east]{$e_1$} -- (2,0) node[anchor=west]{$e_2$} -- (1,1.73205080757) node[anchor=south]{$e_3$} -- cycle;
        \draw[dotted] (-1,0.86602540378) -- (3,0.86602540378) node[anchor=west]{$A(S)\cap\{x_1+x_2+x_3=1\}$};
        \filldraw (0.5,0.86602540378) circle (2pt) node[anchor=-45]{$p^1$};
        \filldraw (1.5,0.86602540378) circle (2pt) node[anchor=-135]{$p^2$};
    \end{tikzpicture}
    \caption{The generators are $p^1$ and $p^2$.}
    \label{fig:algorithm}
\end{figure}
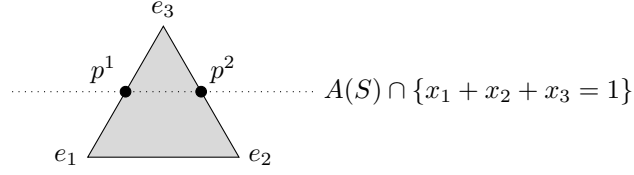
\end{eje}

This process can stop even before, with the help of the following results.
\begin{prop}
    Let $F$ be a closed face of a compact convex set $K$ such that $F\neq\emptyset$ and $G\subsetneq F$ face of $K$ implies $G=\emptyset$. Then $F$ is a singleton.
\end{prop}
\begin{proof}
    By the Krein-Milman theorem, $F$ must be the convex hull of its extreme points. But the faces of $F$ are faces of $K$. Therefore, $F$ must be a singleton.
\end{proof}
This proposition guarantees that when considering a $k$-simplex $H$ in step $k+1$, if a solution of \eqref{mart_sign} exists in $H$, then it must be unique.

\begin{prop}\label{dim_sing}
    Let $K\subset\mathbb{R}^b$ be a convex set. Let $F$ be a face of $K$ such that $F\cap A(S)$ is a singleton and $G\subsetneq F$ face of $K$ implies $G\cap A(S)=\emptyset$. Then $\dim F\leq b-\dim A(S)$.
\end{prop}
The proof of this proposition requires some definitions, notations and propositions from \cite{weis2024note}.
\begin{nota}
    Given a set $K\subset\mathbb{R}^b$, we will denote its affine hull by $\text{aff}(K)$.
\end{nota}
\begin{nota}
    Given a convex set $K\subset\mathbb{R}^b$ and $x\in K$, we denote the smallest face of $K$ containing $x$ by $F_K(x)$.
\end{nota}
\begin{defi}
    Let $K\subset\mathbb{R}^b$ be a convex set and $x\in K$. $x$ is an \emph{internal point} of $K$ if for all $y\in K$ there exists $\varepsilon>0$ such that $x+\varepsilon (x-y)\in K$.
\end{defi}
\begin{defi}
    Let $K\subset\mathbb{R}^b$ be a convex set and $x\in K$, we define
    \begin{align*}
        C_K(x)=\{y\in\mathbb{R}^b|\exists\varepsilon>0:x+\varepsilon (y-x)\in K\}\text{ .}
    \end{align*}
\end{defi}
\begin{prop}
    Let $K\subset\mathbb{R}^b$ be a convex set. Every point $x\in K$ is an internal point of $F_K(x)$.
\end{prop}
\begin{prop}
    Let $K\subset\mathbb{R}^b$ be a convex set. For all $x\in K$, the following assertions are equivalent.
    \begin{enumerate}
        \item The point $x$ is an internal point of $K$.
        \item We have $C_K(x)=\text{aff}(K)$.
    \end{enumerate}
\end{prop}
\begin{proof}[Proof of proposition \ref{dim_sing}]
    Let $x\in F\cap A(S)$. Then $A(S)-x$ and $\text{aff}(F)-x$ are linear subspaces of $\mathbb{R}^b$ and
    \begin{align*}
        \dim((A(S)-x)+(\text{aff}(F)-x))+\dim((A(S)-x)\cap(\text{aff}(F)-x))\\
        =\dim(A(S)-x)+\dim(\text{aff}(F)-x)\text{ .}
    \end{align*}
    That is to say,
    \begin{align*}
        \dim F&=\dim((A(S)-x)+(\text{aff}(F)-x))+\dim((A(S)-x)\cap(\text{aff}(F)-x))-\dim A(S)\\
        &\leq b+\dim(A(S)\cap \text{aff}(F))-\dim A(S)\text{ .}
    \end{align*}
    On the other hand, observe that $F_K(x)\subset F$. As a matter of fact, $F_K(x)=F$: otherwise $x\in F_K(x)\cap A(S)=\emptyset$. Therefore $x$ is an internal point of $F$, hence $\text{aff}(F)=C_F(x)$.
    
    Let $y\in A(S)\cap\text{aff}(F)$. Then there exists $\varepsilon>0$ such that $z:=x+\varepsilon (y-x)\in F\Rightarrow z\in F\cap A(S)\Rightarrow z=x\Rightarrow y=x$. Thus $\dim(A(S)\cap \text{aff}(F))=0$.
\end{proof}
From proposition \ref{dim_sing}, we can deduce that it is unnecessary to check if the convex hulls generated by $b+2-\dim A(S)$ or more points intersect $A(S)$.

Finally, the generators obtained by this algorithm are necessary. This follows from the next easy to prove proposition, stated in \cite{narici2010topological}.
\begin{prop}
    Let $K$ be a convex set and $x\in K$. $x$ is an extreme point of K if and only if $K\backslash\{x\}$ is convex.
\end{prop}

\begin{eje}
The generators obtained by the algorithm may not be affinely independent. Let's consider the system
\begin{align*}
    \begin{pmatrix}
        1 & -1 & -1 & 1 & 0 & 0\\
        1 & -3 & -2 & 0 & -2 & 0\\
        1 & 1 & 2 & 0 & 0 & 2
    \end{pmatrix}p=\begin{pmatrix}
        0\\
        -1\\
        1
    \end{pmatrix}\text{ .}
\end{align*}
The generators in this case turn out to be
\begin{align*}
\begin{pmatrix}
    1/2 & 1/2 & 0 & 0 & 0 & 0
\end{pmatrix}\text{ ,}\\
\begin{pmatrix}
    0 & 0 & 1/2 & 1/2 & 0 & 0
\end{pmatrix}\text{ ,}\\
\begin{pmatrix}
    0 & 0 & 0 & 0 & 1/2 & 1/2
\end{pmatrix}\text{ ,}\\
\begin{pmatrix}
    1/3 & 0 & 1/3 & 0 & 1/3 & 0
\end{pmatrix}\\
\text{and }\begin{pmatrix}
    0 & 1/3 & 0 & 1/3 & 0 & 1/3
\end{pmatrix}\text{ ,}
\end{align*}
and these generators are not affinely independent, because the average of the first three ones equals the average of the last two ones. Using the algorithm, we can easily check that none of these generators is a convex combination of the others.
\end{eje}

Once we get the generators $p^1,p^2,\ldots,p^k$, the MM's are convex combinations of them, that is to say of the form $\alpha_1 p^1+\ldots+\alpha_k p^k$, or in matrix form
\begin{align*}
    \begin{pmatrix}
        p^1_1&\cdots&p^k_1\\
        \vdots&&\vdots\\
        p^1_b&\cdots&p^k_b
    \end{pmatrix}\begin{pmatrix}
        \alpha_1\\
        \vdots\\
        \alpha_k
    \end{pmatrix}\text{ .}
\end{align*}
For them to be equivalent, they have to satisfy $\alpha_1 p^1_i+\ldots+\alpha_k p^k_i\neq 0$ for all $i$, that is to say that $(\alpha_1,\ldots,\alpha_k)\notin\langle (p^1_i,\ldots,p^k_i)\rangle^{\perp}$ for all $i$. Taking advantage of the fact that $p^j_i\geq 0$ for all $j$, the condition over $\alpha_j$ becomes clear: $\alpha_j>0$ for some $j$ such that $p^j_i\neq 0$.\newline

When we want to add assets to complete the market, let's say $S_{n+1},\ldots,S_m$, to keep it arbitrage-free we have the restriction
\begin{align*}
    \begin{pmatrix}
        Z_{n+1}(1,1)&\cdots&Z_{n+1}(b,1)\\
        \vdots&&\vdots\\
        Z_{m}(1,1)&\cdots&Z_{m}(b,1)
    \end{pmatrix}\begin{pmatrix}
        p^1_1&\cdots&p^k_1\\
        \vdots&&\vdots\\
        p^1_b&\cdots&p^k_b
    \end{pmatrix}\begin{pmatrix}
        \alpha_1\\
        \vdots\\
        \alpha_k
    \end{pmatrix}=\begin{pmatrix}
        Z_{n+1}(0)\\
        \vdots\\
        Z_{m}(0)
    \end{pmatrix}\text{ ,}
\end{align*}
where the values of $(\alpha_1,\ldots,\alpha_k)$ have the restrictions previously mentioned. Once chosen the values of $S_{l}(\omega,1)$ to complete the market, the range of values for $(S_{n+1}(0),\ldots,S_m(0))$ is given by
\begin{align*}
    \begin{pmatrix}
        S_{n+1}(0)\\
        \vdots\\
        S_{m}(0)
    \end{pmatrix}=S_0(0)
    \begin{pmatrix}
        \langle Z_{n+1}(1),p^1\rangle &\cdots&\langle Z_{n+1}(1),p^k\rangle\\
        \vdots&&\vdots\\
        \langle Z_{m}(1),p^1\rangle &\cdots&\langle Z_{m}(1),p^k\rangle
    \end{pmatrix}\begin{pmatrix}
        \alpha_1\\
        \vdots\\
        \alpha_k
    \end{pmatrix}\text{ ,}
\end{align*}
where the values of $(\alpha_1,\ldots,\alpha_k)$ have the restrictions previously mentioned.

\subsection{More trading dates}

In a market with a finite number of trading dates and a finite number of forks, the analysis of viability and completeness can be carried out with the tools we have developed. Let's suppose the trading dates are $t=0,1,\ldots,T$. In such a market, the sample space $\Omega$ is finite and there is a family $\{\mathcal{F}_t\}_{t=0}^T$ of $\sigma$-algebras over $\Omega$, that represent the information at each instant of time: if the state of the world is $\omega$ and $A\in\mathcal{F}_t$, at time $t$ we can say if $\omega\in A$ or not. Here, as the sample space is finite, each $\sigma$-algebra over $\Omega$ corresponds to a partition over $\Omega$, so there is a family $\{\mathcal{P}_t\}_{t=0}^T$ of partitions over $\Omega$ that represent the information at each instant of time.

Given $t<T$, for each $A\in\mathcal{P}_t$, we can consider the submarket with trading dates $t$ and $t+1$, and a sample space given by
\begin{align*}
    \Omega_{t,A}=\{B\in\mathcal{P}_{t+1}:B\subset A\}\text{ .}
\end{align*}
We will call such submarket the \emph{$(t,A)$ component} of the \emph{total market}. The total market is arbitrage-free if and only if all of its components are arbitrage-free (similar statements can be found in \cite{delbaen2006mathematics} and \cite{mishura2021discrete}). Analogously, the total market is complete if and only if all of its components are complete. So we can use the previously developed tools to analyze if all the components are arbitrage-free and complete, and from that deduce the viability and completeness of the total market.

If the total market is arbitrage-free but not complete, then we can complete it by completing all of its components. Every asset of the $(t,A)$ component can be thought as an asset of the total market that trades in dates $t$ and $t+1$ in the case that $\omega\in A$. So, every asset we add to complete the components is an asset we add to complete the total market. In the literature, assets are considered to be traded the whole time, from date $0$ to $T$; assets traded in dates $t$ and $t+1$ in the case that $\omega\in A$ can be extended by investing in $S_0$ between other dates and in other scenarios, which doesn't change anything else.

\section{Concrete examples}\label{eje}

\begin{eje}
Let's consider the system
\begin{align*}
\begin{pmatrix}
18 & -6 & -6 & 75\\
99 & -33 & -33 & 291
\end{pmatrix}p=\begin{pmatrix}
15\\
123
\end{pmatrix}\text{ .}
\end{align*}
Applying the method we don't get any generator, which indicates that there are no solutions of the system in the simplex. Nevertheless, it is immediate that each row by itself has solutions in the simplex, and $(2,-1/3,-1/3,-1/3)$ is a solution in the hyperplane $x_1+x_2+x_3+x_4=1$.
\end{eje}

\begin{eje}
Let's consider the system
\begin{align*}
\begin{pmatrix}
-3&   1& -15&   1\\
-3&   1& -7&   1
\end{pmatrix}p=\begin{pmatrix}
-3\\
-3
\end{pmatrix}\text{ .}
\end{align*}

Applying the method we get vector $(1,0,0,0)$ as sole generator, which indicates that the solution of the system in the simplex is just that vertex. Nevertheless, it is immediate that each row by itself has more solutions in the simplex, and $(1,1,0,-1)$ is another solution in the hyperplane $x_1+x_2+x_3+x_4=1$.\newline
The only MM is then $(1,0,0,0)$, but it is not equivalent.
\end{eje}

\begin{eje}
Let's consider the system
\begin{align*}
\begin{pmatrix}
-1&   -1& -3&   3\\
1&   1& -3&   3
\end{pmatrix}p=\begin{pmatrix}
-1\\
1
\end{pmatrix}\text{ .}
\end{align*}

Applying the method we get vectors $(1,0,0,0)$ and $(0,1,0,0)$ as the generators, which indicates that the solution set of the system in the simplex is the edge that connects these vertices. Nevertheless, it is immediate that each row by itself has more solutions in the simplex.\newline
The MM's are the convex combinations of $(1,0,0,0)$ and $(0,1,0,0)$, but none of them is equivalent, because in all of these combinations the third and fourth components are null.
\end{eje}

\begin{eje}
Let's consider the system
\begin{align*}
\begin{pmatrix}
2&   0& 0&   0
\end{pmatrix}p=
1\text{ .}
\end{align*}

Applying the method we get vectors $(1/2,1/2,0,0)$, $(1/2,0,1/2,0)$ and $(1/2,0,0,1/2)$ as the generators. The MM's are the convex combinations of these vectors, that is to say
\begin{align*}
    \alpha_1(1/2,1/2,0,0)+\alpha_2(1/2,0,1/2,0)+(1-\alpha_1-\alpha_2)(1/2,0,0,1/2)
\end{align*}
with $0\leq\alpha_1$, $0\leq\alpha_2$ and $\alpha_1+\alpha_2\leq 1$. To find the equivalent ones, we observe that
\begin{enumerate}
    \item From the first components of each generator, no one is null, so the first component of the combination is not null if and only if $\alpha_1>0$, $\alpha_2>0$ or $\alpha_1+\alpha_2<1$.
    \item From the second components of each generator, just the one from the first generator is not null, so the second component of the combination is not null if and only if $\alpha_1>0$.
    \item From the third components of each generator, just the one from the second generator is not null, so the third component of the combination is not null if and only if $\alpha_2>0$.
    \item From the fourth components of each generator, just the one from the third generator is not null, so the fourth component of the combination is not null if and only if $\alpha_1+\alpha_2<1$.
\end{enumerate}

Then the EMM's are the combinations with $0<\alpha_1$, $0<\alpha_2$ and $\alpha_1+\alpha_2<1$.\newline
To complete the market we can add two assets, extending the matrix to
\begin{align*}
\begin{pmatrix}
1&   1& 1&   1\\
2&   0& 0&   0\\
0&   1& 0&   0\\
0&   0& 1&   0
\end{pmatrix}\text{ .}
\end{align*}
Multiplying the last two rows of this matrix with the matrix of generators
\begin{align*}
\begin{pmatrix}
1/2& 1/2& 1/2\\
1/2&   0&   0\\
0&   1/2&   0\\
0&     0& 1/2
\end{pmatrix}
\end{align*}
we get the matrix
\begin{align*}
\begin{pmatrix}
1/2&   0& 0\\
0&   1/2& 0
\end{pmatrix}\text{ .}
\end{align*}
So the extended system would be
\begin{align*}
\begin{pmatrix}
2&   0& 0&   0\\
0&   1& 0&   0\\
0&   0& 1&   0
\end{pmatrix}p=\begin{pmatrix}
1\\
\alpha_1/2\\
\alpha_2/2
\end{pmatrix}\text{ ,}
\end{align*}
with the restrictions $\alpha_1>0$, $\alpha_2>0$ and $\alpha_1+\alpha_2<1$. For example, taking $\alpha_1=\alpha_2=1/3$, we get the system
\begin{align*}
\begin{pmatrix}
2&   0& 0&   0\\
0&   1& 0&   0\\
0&   0& 1&   0
\end{pmatrix}p=\begin{pmatrix}
1\\
1/6\\
1/6
\end{pmatrix}\text{ .}
\end{align*}

Using the algorithm we can easily check that the only generator is $(1/2,1/6,1/6,1/6)$. Therefore, that is the only MM, and none of its components are null so it is equivalent. Thus, this market is arbitrage-free and complete.
\end{eje}

\section{Application}\label{apli}

The simplest model we can consider is the case where there is only one asset. In that case, it is immediate that $A(S)=\mathbb{R}^{b}\Rightarrow\mathcal{M}^e(S)=\Delta^{b-1}$ and thus:
\begin{itemize}
    \item the market is arbitrage-free;
    \item the market is complete if and only if $b=1$.
\end{itemize}

Let's assume that there exists $r\in\mathbb{R}$ such that $S_0(t+1)/S_0(t)=1+r$ for all $t$, and let's add another asset $S_1$ with the following features. $S_1(0)>0$, we have a set of coefficients $0<f_1<f_2<\ldots<f_b$ and in each instant of time $t$, we have that
\begin{align*}
\mathbb{P}(S_1(t+1)=f_hS_1(t))>0\,\forall h\qquad\text{and}\qquad\sum_{h=1}^b\mathbb{P}(S_1(t+1)=f_hS_1(t))=1\text{ .}
\end{align*}
Let's fix the variable $t$. From what was previously seen, the market is arbitrage-free if and only if $Z_1(t)$ is in the open convex hull generated by $\{Z_1(\omega,t+1)\}_{\omega}$, that is to say, if
\begin{align*}
    \min(\{Z(\omega,t+1)\}_{\omega})&<Z(t)<\max(\{Z(\omega,t+1)\}_{\omega})\\
    \Leftrightarrow f_1&<1+r<f_b\text{ .}
\end{align*}
Besides, the market is complete if and only if it is arbitrage-free and
\begin{align*}
    \begin{pmatrix}
            1&\cdots &1\\
            f_1&\cdots &f_b
        \end{pmatrix}
\end{align*}
has row rank $b$, which happens only if $b\leq 2$.\newline
In the case $b=3$, we need the row rank to be $3$, so we can add a derivative to complete the market:
\begin{align*}
    \begin{pmatrix}
            1&1&1\\
            f_1&f_2&f_3\\
            c_1&c_2&c_3
        \end{pmatrix}
\end{align*}
must have rank $3$. Triangulating this matrix we get
\begin{align*}
    \begin{pmatrix}
            1&1&1\\
            0&f_2-f_1 & f_3-f_1\\
            0&0&c_3-c_1-\frac{c_2-c_1}{f_2-f_1}(f_3-f_1)
        \end{pmatrix}\text{ ,}
\end{align*}
so the condition for the row rank to be 3 is
\begin{align}
    c_3-c_1-\frac{c_2-c_1}{f_2-f_1}(f_3-f_1)&\neq 0\nonumber\\
    \Leftrightarrow c_1(f_3-f_2)+c_2(f_1-f_3)+c_3(f_2-f_1)&\neq 0\text{ .}\label{comp_3}
\end{align}

Let's compute the EMM's.

In the case $b=3$, the market is arbitrage-free if $f_1<1+r<f_3$; we will assume these inequalities. The system is
\begin{align*}
    \begin{pmatrix}
            f_1&f_2&f_3
        \end{pmatrix}p=1+r\text{ .}
\end{align*}
In that case, neither $(1,0,0)$ nor $(0,0,1)$ are solutions, but
\begin{align*}
    \left(\frac{f_3-(1+r)}{f_3-f_1},0,\frac{1+r-f_1}{f_3-f_1}\right)
\end{align*}

is a solution.\newline
If $f_2=1+r$, then $(0,1,0)$ is a solution and every solution that is a probability measure is a convex combination of these. Then, the MM's are obtained as
\begin{align*}
    \begin{pmatrix}
        \frac{f_3-(1+r)}{f_3-f_1}&0\\
        0&1\\
        \frac{1+r-f_1}{f_3-f_1}&0
    \end{pmatrix}\begin{pmatrix}
        \alpha_1\\
        \alpha_2
    \end{pmatrix}\text{ .}
\end{align*}

From the first and third rows we get that $\alpha_1>0$, and from the second row we get that $\alpha_2>0$. Then, the EMM's are
\begin{align*}
    \left(p\frac{f_3-(1+r)}{f_3-f_1},1-p,p\frac{1+r-f_1}{f_3-f_1}\right)\text{ ,}
\end{align*}
with $0<p<1$. To complete the market, we have that $c_0$, the derivative's price at time $t$, is a convex combination
\begin{align*}
    \alpha_1\frac{1}{1+r}\begin{pmatrix}
        c_1&c_2&c_3
    \end{pmatrix}\begin{pmatrix}
        \frac{f_3-(1+r)}{f_3-f_1}\\
        0\\
        \frac{1+r-f_1}{f_3-f_1}
    \end{pmatrix}+\alpha_2\frac{1}{1+r}\begin{pmatrix}
        c_1&c_2&c_3
    \end{pmatrix}\begin{pmatrix}
        0\\
        1\\
        0
    \end{pmatrix}\\
    =\alpha_1\frac{c_1(f_3-(1+r))+c_3(1+r-f_1)}{(f_3-f_1)(1+r)}+\alpha_2\frac{c_2}{1+r}\text{ ,}
\end{align*}
where $\alpha_1>0$ and $\alpha_2>0$. That is to say, $c_0$ is in the open interval between
\begin{align*}
    \frac{c_1(f_3-(1+r))+c_3(1+r-f_1)}{(f_3-f_1)(1+r)}
\end{align*}
and
\begin{align*}
    \frac{c_2}{1+r}\text{ .}
\end{align*}

If $f_2<1+r$, then $(0,1,0)$ is not a solution either, but
\begin{align*}
    \left(0,\frac{f_3-(1+r)}{f_3-f_2},\frac{1+r-f_2}{f_3-f_2}\right)
\end{align*}

is a solution. There are no solutions of the form $(p,1-p,0)$, so the MM's are obtained as
\begin{align*}
    \begin{pmatrix}
        \frac{f_3-(1+r)}{f_3-f_1}&0\\
        0&\frac{f_3-(1+r)}{f_3-f_2}\\
        \frac{1+r-f_1}{f_3-f_1}&\frac{1+r-f_2}{f_3-f_2}
    \end{pmatrix}\begin{pmatrix}
        \alpha_1\\
        \alpha_2
    \end{pmatrix}\text{ .}
\end{align*}

From the first row we get that $\alpha_1>0$ and from the second row that $\alpha_2>0$. The third row does not impose any additional restriction. Then, the EMM's are
\begin{align}
    \left(p\frac{f_3-(1+r)}{f_3-f_1},(1-p)\frac{f_3-(1+r)}{f_3-f_2},p\frac{1+r-f_1}{f_3-f_1}+(1-p)\frac{1+r-f_2}{f_3-f_2}\right)\text{ ,}\label{emm2}
\end{align}
with $0<p<1$. To complete the market, we have that $c_0$ is a convex combination
\begin{align*}
    \alpha_1\frac{c_1(f_3-(1+r))+c_3(1+r-f_1)}{(f_3-f_1)(1+r)}+\alpha_2\frac{c_2(f_3-(1+r))+c_3(1+r-f_2)}{(f_3-f_2)(1+r)}\text{ ,}
\end{align*}
where $\alpha_1>0$ and $\alpha_2>0$. That is to say, $c_0$ is in the open interval between
\begin{align*}
    \frac{c_1(f_3-(1+r))+c_3(1+r-f_1)}{(f_3-f_1)(1+r)}
\end{align*}
and
\begin{align*}
    \frac{c_2(f_3-(1+r))+c_3(1+r-f_2)}{(f_3-f_2)(1+r)}\text{ .}
\end{align*}

If $f_2>1+r$, then $(0,1,0)$ is not a solution either, but
\begin{align*}
    \left(\frac{f_2-(1+r)}{f_2-f_1},\frac{1+r-f_1}{f_2-f_1},0\right)
\end{align*}

is a solution. There are no solutions of the form $(0,p,1-p)$, so the MM's are obtained as
\begin{align*}
    \begin{pmatrix}
        \frac{f_3-(1+r)}{f_3-f_1}&\frac{f_2-(1+r)}{f_2-f_1}\\
        0&\frac{1+r-f_1}{f_2-f_1}\\
        \frac{1+r-f_1}{f_3-f_1}&0
    \end{pmatrix}\begin{pmatrix}
        \alpha_1\\
        \alpha_2
    \end{pmatrix}\text{ .}
\end{align*}
From the second row we get that $\alpha_2>0$ and from the third row that $\alpha_1>0$. The first row does not impose any additional restriction. To complete the market, we have that $c_0$ is a convex combination
\begin{align*}
    \alpha_1\frac{c_1(f_3-(1+r))+c_3(1+r-f_1)}{(f_3-f_1)(1+r)}+\alpha_2\frac{c_1(f_2-(1+r))+c_2(1+r-f_1)}{(f_2-f_1)(1+r)}\text{ ,}
\end{align*}
where $\alpha_1>0$ and $\alpha_2>0$. That is to say, $c_0$ is in the open interval between
\begin{align*}
    \frac{c_1(f_3-(1+r))+c_3(1+r-f_1)}{(f_3-f_1)(1+r)}
\end{align*}
and
\begin{align*}
    \frac{c_1(f_2-(1+r))+c_2(1+r-f_1)}{(f_2-f_1)(1+r)}\text{ .}
\end{align*}

\subsection{The Korn–Kreer–Lenssen model}

The Korn-Kreer-Lenssen model \cite{chenkorn,korn1998pricing} is a continuous-time model characterized by an asset with a price process given by
\begin{align*}
    \mathbb{P}(dS=1)&=\lambda S\,dt\text{ ,}\\
    \mathbb{P}(dS=-1)&=\eta S\,dt\\
    \text{and}\quad\mathbb{P}(dS=0)&=1-(\lambda+\eta)S\,dt\text{ ,}
\end{align*}
for some fixed $\lambda,\eta>0$. Trading is continuous in $[0,T]$, and the interest rate $r$ is constant. In order to complete it, a put option $F^*$ with strike $1$ and maturity $T$ is added.

By partitioning $[0,T]$ in $n$ equal intervals, we get a finite version of the KKL model:
\begin{align*}
    \mathbb{P}(S(t+\Delta t)=S(t)+1)&=\lambda S(t)\Delta t\text{ ,}\\
    \mathbb{P}(S(t+\Delta t)=S(t)-1)&=\eta S(t)\Delta t\\
    \text{and}\quad\mathbb{P}(S(t+\Delta t)=S(t))&=1-(\lambda+\eta)S(t)\Delta t\text{ ,}
\end{align*}
where $\Delta t:=T/n$, and the rate of interest is $r\Delta t$. For these probabilities to make sense when $n\to\infty$, the condition $(\lambda+\eta)T<1$ must be required.

If $S(t)=0$, we fall under the case $b=1$, with a trivially unique EMM. If not, we fall under the case $b=3$ with
\begin{align*}
    f_1=1-S(t)^{-1}\text{ ,}\qquad f_2=1\qquad\text{and}\qquad f_3=1+S(t)^{-1}\text{ .}
\end{align*}
Therefore, the market is arbitrage-free if and only if
\begin{align*}
    1-S(t)^{-1}<1+r\Delta t<1+S(t)^{-1}\\
    \Leftrightarrow\Delta t<\frac{1}{|r|S(t)}\text{ .}
\end{align*}
For this to hold for each $t<T$, it is necessary and sufficient to hold under the case $S(t)=S(0)+n-1$. Therefore, the condition translates to
\begin{align*}
    \Delta t&<\frac{1}{|r|(S(0)+n-1)}\\
    \Leftrightarrow T|r|(S(0)+n-1)&<n\text{ .}
\end{align*}
This holds for $n$ sufficiently large if $T|r|<1$.\newline
To complete the market, another asset has to be added. If at time $t$ the original asset has price $S(t)$, then we will assume the new asset has price $F^n(t,S(t))$ for some function $F^n$. Let $\Gamma\subset (T/n)\{0,1,\ldots,n\}\times\mathbb{N}_0$ be the smallest set that satisfies:
\begin{enumerate}
    \item $(0,S(0))\in\Gamma$
    \item If $(t,k)\in\Gamma$ and $t<T$, then
    \begin{itemize}
        \item $(t+\Delta t,0)\in\Gamma$ if $k=0$.
        \item $(t+\Delta t,k-1),(t+\Delta t,k),(t+\Delta t,k+1)\in\Gamma$ if $k\neq 0$.
    \end{itemize}
\end{enumerate}
For each $(t,k)\in\Gamma$ such that $t<T$ and $k\neq 0$, as seen in this section, we have several EMM's for the transitions
\begin{align*}
    \tilde{\mathbb{P}}(S(t+\Delta t)=k+i|S(t)=k)
\end{align*}
for $i\in\{-1,0,1\}$. Let's choose any of these EMM's for each of these $(t,k)$. The chosen EMM determines the values of $F^n(t,k)$ for $t<T$, given the values of $F^n(T,k)$:
\begin{align}
    F^n(t,k)=\begin{cases}
        (1+r\Delta t)^{-1}\sum_{i=-1}^1\tilde{\mathbb{P}}(S(t+\Delta t)=k+i|S(t)=k)F^n(t+\Delta t,k+i)&\text{ if }k\neq 0\\
        (1+r\Delta t)^{-1}F^n(t+\Delta t,0)&\text{ if }k=0
    \end{cases}\text{ .}
\end{align}
In order for the market to be complete, for each $(t,k)\in\Gamma$ such that $t<T$ and $k\neq 0$, equation \eqref{comp_3} requires
\begin{align}
    F^n(t+\Delta t,k-1)k^{-1}-2F^n(t+\Delta t,k)k^{-1}+F^n(t+\Delta t,k+1)k^{-1}&\neq 0\nonumber\\
    \Leftrightarrow F^n(t+\Delta t,k-1)-2F^n(t+\Delta t,k)+F^n(t+\Delta t,k+1)&\neq 0\text{ .}\label{comp_3_2}
\end{align}
In \cite{korn1998pricing} the new asset is assumed to be a put option with strike $1$ and maturity $T$. This corresponds to
\begin{align*}
    F^*(T,k)=\begin{cases}
        0&\text{if }k\geq 1\\
        1&\text{if }k=0
    \end{cases}\text{ .}
\end{align*}
If we use these values in our discrete-time version, then the market will not be complete.
\begin{prop}
    Given any values $F^*(T,k)$, the values of $F^n(T,k)$ can be taken arbitrarily close to $F^*(T,k)$ while also getting a complete market.
\end{prop}
\begin{proof}
    Let $\tilde{\Gamma}$ be the set
    \begin{align*}
        \tilde{\Gamma}:=\{(t,k)\in\Gamma:t<T,k\neq 0\}\text{ .}
    \end{align*}
    Observe that each $F^n(t,k)$ is a linear combination of values $F^n(T,k)$ through the already chosen EMM. For each $(t,k)\in\tilde{\Gamma}$, let $H_{(t,k)}$ be the set
    \begin{align*}
        H_{(t,k)}:=\{(F^n(T,l))_{l=0}^{S(0)+n}:F^n(t+\Delta t,k-1)-2F^n(t+\Delta t,k)+F^n(t+\Delta t,k+1)=0\}\\
        \subset\mathbb{R}^{S(0)+n+1}\text{ .}
    \end{align*}
    These sets are either hyperplanes or $\mathbb{R}^{S(0)+n+1}$. For a given $(t,k)$, $H_{(t,k)}=\mathbb{R}^{S(0)+n+1}$ if and only if the coefficients corresponding to each $F^n(T,l)$ in the linear combination
    \begin{align*}
        F^n(t+\Delta t,k-1)-2F^n(t+\Delta t,k)+F^n(t+\Delta t,k+1)
    \end{align*}
    are each and everyone of them $0$. This is not true, because the coefficient corresponding to $F^n(T,k+(T-t)/\Delta t)$ is not $0$:
    \begin{align*}
        F^n(t+\Delta t,k+i)=\left(1+r\Delta t\right)^{-\frac{T-t-\Delta t}{\Delta t}}\sum_j\tilde{\mathbb{P}}(S(T)=j|S(t+\Delta t)=k+i)F^n(T,j)
    \end{align*}
    and
    \begin{align*}
        \tilde{\mathbb{P}}(S(T)=j|S(t+\Delta t)=k+i)=0
    \end{align*}
    for $|k+i-j|>(T-t-\Delta t)/\Delta t$. Therefore with $j=k+(T-t)/\Delta t$,
    \begin{align*}
        \tilde{\mathbb{P}}(S(T)=j|S(t+\Delta t)=k+i)=0
    \end{align*}
    for $i\in\{-1,0\}$, but
    \begin{align*}
        \tilde{\mathbb{P}}(S(T)=j|S(t+\Delta t)=k+1)>0\text{ .}
    \end{align*}
    Hence the coefficient corresponding to $F^n(T,k+(T-t)/\Delta t)$ is
    \begin{align*}
        \left(1+r\Delta t\right)^{-\frac{T-t-\Delta t}{\Delta t}}\tilde{\mathbb{P}}(S(T)=k+(T-t)/\Delta t|S(t+\Delta t)=k+1)>0\text{ .}
    \end{align*}
    Thus, the sets $H_{(t,k)}$ are hyperplanes and therefore, they are nowhere dense. Hence, their union is meagre and, as $\mathbb{R}^{S(0)+n+1}$ is a Baire space, its complement is dense.

    Let $\varepsilon>0$ and $(F^*(T,l))_{l=0}^{S(0)+n}\in\mathbb{R}^{S(0)+n+1}$. There is a $(F^n(T,l))_{l=0}^{S(0)+n}$ such that
    \begin{align*}
        (F^n(T,l))_{l=0}^{S(0)+n}\notin\bigcup_{(t,k)\in\tilde{\Gamma}}H_{(t,k)}\\
        \text{and }|F^n(T,l)-F^*(T,l)|<\varepsilon
    \end{align*}
    for all $0\leq l\leq S(0)+n$. In this case, the market is complete because of equation \eqref{comp_3_2}.
\end{proof}

It would be desirable to take the limit $n\to\infty$, and deduce from these results the viability and completeness of the original KKL model with continuous time. Nevertheless, this step is not so simple and requires further research. To illustrate this, we provide the following example.
\begin{eje}
In this example we will show a sequence of complete markets that converge to a limit market with arbitrage opportunities. Each of these markets is characterized by a deterministic bond price process $B^n$ and a stochastic asset price process $S^n$ (the corresponding processes are $B$ and $S$ in the case of the limit market). The convergence concept we prove is given by:
\begin{enumerate}
    \item $B^n\to B$ uniformly;
    \item $S^n(0)\to S(0)$;
    \item and for each $t_1<t_2$, $\log(S^n(t_2)/S^n(t_1))$ converges in distribution to $\log(S(t_2)/S(t_1))$.
\end{enumerate}

For the markets in the sequence, let's consider the classical binomial model. Given $n\in\mathbb{N}$, let's set the trading dates $0,1/n,2/n,\ldots, 1$, with parameters $r_n=0$, $u_n=e$, $d_n=\exp(-1/n^2)$ and $p_n=1/n$, with probabilities
\begin{align*}
    \mathbb{P}(S^n(t+1/n)=u_nS^n(t))=p_n\\
    \text{and }\mathbb{P}(S^n(t+1/n)=d_nS^n(t))=1-p_n\text{ ,}
\end{align*}
and $S^n(0)=B^n(0)=1$. Observe that all these markets are complete, since $d_n<1+r_n<u_n$. The processes $S^n,B^n:\{0,1/n,\ldots, 1\}\to\mathbb{R}$ can be extended to continuous time processes in the interval $[0,1]$ by thinking about them as càdlàg functions constant between trading dates. Hence, we get
\begin{align*}
    \log S^n(t)=\sum_{i=1}^{\lfloor nt\rfloor}X_i^n\text{ ,}
\end{align*}
where $X_i^n$ are independent,
\begin{align*}
    \mathbb{P}(X_i^n=\log u_n)=p_n\\
    \text{and }\mathbb{P}(X_i^n=\log d_n)=1-p_n\text{ .}
\end{align*}
For the limit market, let $Y(t)$ be a Poisson process with parameter $1$, $S(t):=\exp(Y(t))$ and $B(t)\equiv 1$. The portfolio process $S(t)-B(t)$ is an arbitrage opportunity in this market:
\begin{align*}
    S(0)-B(0)=0\text{ ,}\\
    \mathbb{P}(S(1)-B(1)\geq 0)=\mathbb{P}(Y(1)\geq 0)=1\\
    \text{and }\mathbb{P}(S(1)-B(1)> 0)=1-e^{-1}>0\text{ .}
\end{align*}

Let $0\leq t_1<t_2\leq 1$. For the sequence of markets
\begin{align*}
    \log\left(\frac{S^n(t_2)}{S^n(t_1)}\right)=\log(u_n/d_n)Y^n_{t_1,t_2}+(\lfloor n t_2\rfloor-\lfloor n t_1\rfloor)\log d_n\text{ ,}
\end{align*}
where $Y^n_{t_1,t_2}\sim B(\lfloor n t_2\rfloor-\lfloor n t_1\rfloor,p_n)$. The characteristic function of $\log(S^n(t_2)/S^n(t_1))$ is
\begin{align*}
    \varphi_n(s)=\exp\left(is(\lfloor n t_2\rfloor-\lfloor n t_1\rfloor)\log d_n\right)\left(1-p_n+p_n\exp(is\log(u_n/d_n))\right)^{\lfloor n t_2\rfloor-\lfloor n t_1\rfloor}\text{ .}
\end{align*}
For the limit market, $\log(S(t_2)/S(t_1))\sim\text{Pois}(t_2-t_1)$ and its characteristic function is
\begin{align*}
    \varphi(s)=\exp\left((t_2-t_1)(e^{is}-1)\right)\text{ .}
\end{align*}
As $\varphi_n$ converges pointwise to $\varphi$, $\log(S^n(t_2)/S^n(t_1))$ converges in distribution to $\log(S(t_2)/S(t_1))$. Also $S^n(0)=S(0)$ and $B^n\equiv B$, and these facts taken together show that the sequence of markets converges to the limit market.
\end{eje}

\section{Conclusions}

In this work we studied the set of equivalent martingale measures (EMM's), for finite markets. We did this as a way to understand under which conditions those markets are arbitrage-free or complete, and how they can be completed.

We used a geometric interpretation of the set of EMM's to reduce the problem to one of convex geometry, and thus we observed that all the martingale measures can be described as convex combinations of a few of them. We presented an algorithm that allows to compute these generators. We also observed that those generators turn out to be necessary and sufficient.\newline
We presented a simple characterization of the EMM's based on the generators thus obtained, and an algorithm with all the possible ways of completing an arbitrage-free market. Examples of application to concrete cases were provided.

Further results were obtained for the particular cases $b=2$ and $b=3$, with a single asset. We found a necessary and sufficient condition for the market to be arbitrage-free for any value of $b$. Furthermore, we proved that the market is complete when $b=2$. In the case $b=3$, we added a derivative and characterized all the derivatives that complete the market. We used these results to analyze a discrete-time version of the Korn-Kreer-Lenssen model, giving a necessary and sufficient condition for the market to be arbitrage-free, and a way to complete the market in the spirit of \cite{korn1998pricing}.

Further research is required to close the gap between these discrete-time models and continuous-time models. Besides the example given in this paper, particular examples of taking limits on sequences of markets can be found in \cite{mishura2021discrete} and \cite{pakkanen2010microfoundations}, but a general theory seems to be lacking.

\bibliographystyle{plain}
\bibliography{main}

\end{document}